\documentclass[journal]{IEEEtran}
\ifCLASSINFOpdf
\else
\fi
\usepackage{subfigure} 
\usepackage{graphicx}
\usepackage{amssymb}
\usepackage{amsmath}
\usepackage{color} 
\usepackage{cite} 
\newtheorem{theorem}{Theorem}

\newtheorem{lemma}{Lemma}
\newtheorem{definition}{Definition}
\newtheorem{assumption}{Assumption}
\newtheorem{remark}{Remark}
\newtheorem{proof}{Proof}
\usepackage{tikz,xcolor,hyperref}
\definecolor{lime}{HTML}{A6CE39}
\DeclareRobustCommand{\orcidicon}{
	\begin{tikzpicture}
		\draw[lime, fill=lime] (0,0) 
		circle [radius=0.16] 
		node[white] {{\fontfamily{qag}\selectfont \tiny ID}}; 
		\draw[white, fill=white] (-0.0625,0.095) 
		circle [radius=0.007];	  
	\end{tikzpicture}
	\hspace{-2mm}}
\foreach \x in {A, ..., Z}{
	\expandafter\xdef\csname orcid\x\endcsname{\noexpand\href{https://orcid.org/\csname orcidauthor\x\endcsname}{\noexpand\orcidicon}}
}
\makeatother
\begin{document}
\title{A Marginal Cost Consensus Scheme with Reset Mechanism for Distributed Economic Dispatch in BESSs}
\newcommand{\orcidauthorA}{0000-0002-3788-700X}
\newcommand{\orcidauthorB}{0000-0002-3565-4800}
\newcommand{\orcidauthorC}{0000-0002-1415-4073}
\author{Yalin Zhang\orcidA{},
        Zhongxin Liu\orcidB{},~\IEEEmembership{Member,~IEEE,}
        and Zengqiang Chen\orcidC{}
\thanks{Manuscript received XX, XX; revised XX, XX;
	accepted XX, XX. This work is supported by the Tianjin Natural Science Foundation of China (Grant No.20JCYBJC01060) , and the National Natural Science Foundation of China (Grant No. 62103203, 61973175) and the General Terminal IC Interdisciplinary Science Center of Nankai University (Corresponding author: Zhongxin Liu.). 
	\par The authors are with the College of Artificial Intelligence, Nankai University, Tianjin 300350, and also with 
	the Tianjin Key Laboratory of Brain Intelligent Rehabilitation, Nankai University, Tianjin 300350, China (e-mail: zhangyl@mail.nankai.edu.cn; lzhx@nankai.edu.cn; chenzq@nankai.edu.cn).}}

\markboth{IEEE TRANSACTIONS ON SMART GRID, VOL. XX, NO. XX, XX XX}%
{Zhang \MakeLowercase{\textit{et al.}}: A Marginal Cost Consensus Scheme with Reset Mechanism for Distributed ED in BESSs}

\maketitle

\begin{abstract}
	Battery energy storage systems (BESSs) are often integrated into the smart grid as the key equipment for valley filling and peak suppression. However, the internal power consumption and capacity degradation of battery cells can not be ignored. In addition, there are certain electric trading between the owner of the grid-connected BESSs and the electric company (EC). Therefore, an expenditure function for grid-connected BESSs is constructed in this paper, in which internal power consumption, capacity degradation and power trading are covered while meeting the balance of power supply and demand. On this basis, the Karush-Kuhn-Tucker (KKT) condition of the function is summarized as the consensus problem of marginal cost (MC) converging to time-phased electricity price. Thus, we plan to design a distributed MC consensus scheme for economic dispatch (ED) in BESSs based on multi-agent systems (MASs) on a small-time scale. In view of the low control accuracy and response speed of the existing distributed proportional protocol, this paper proposes a distributed ED (DED) scheme with reset mechanism based on a proportional integral (PI) control. When the proportional term encounters zero crossing, the integral term of the control scheme is reset to 0, which ensures that the signs of the two are aligned, thus accelerating MCs convergence and restraining overshoot. Parameter conditions for consensus, regularity and Zeno-free behavior are given through the relevant theoretical analysis. Several simulation cases are designed to verify the designed DED algorithm.
\end{abstract}

\begin{IEEEkeywords}
Battery energy storage systems (BESSs), distributed economic dispatch (DED), marginal cost (MC), multi-agent systems (MASs), reset mechanism.
\end{IEEEkeywords}

\IEEEpeerreviewmaketitle
\begin{center}
	N\footnotesize{OMENCLATURE}
\end{center}
\normalsize
\begin{tabbing}
	\hspace{2cm} \= \kill
	BESS\> Battery energy storage system\\
	KKT\>  The Karush-Kuhn-Tucker\\
	MC\> Marginal cost\\
	MAS\> Multi-agent system\\
	ED\> Economic dispatch\\
	DED\> Distributed ED\\
	MG\> Microgrid\\
	UG\> Utility grid\\
	EC\> Electric company\\
	PC\> Proportional controller\\
	PI\> Proportional-integral\\
	ER\> Energy router\\
	$C^P$, $\rho$\> The income and electric price\\
	$P_i$, $P_e$\> The output power and the exchange power\\
	$\pi_i$, $\alpha_i$, $\beta_i$\> The inner power loss ratio and its coefficients\\
	$f_{Li}$\> The capacity deterioration cost\\
	$m_i$\> The coefficient of $f_{Li}$\\
	$F_c$, $D$\> The total expenditure function and load\\
	$L_{ac}$, $\gamma$\> A Lagrangian function and its multiplier\\
	$\lambda_i$, $\lambda$\> MC of BESS $i$ and its stack vector\\
	$u_i$, $u$\> The control input and its stack vector\\
	$\cal G$\> Communication graph\\
	$\cal V$, $\cal E$\> The vertex set and edge set of $\cal G$\\
	$N_i$, $d_i$\> Neighbor set and count of agent $i$\\
	$\cal A$, $\cal D$\> The adjacency matrix and in-degree matrix\\
	$b$, $b_{ii}$\> The adjacency vector and its component\\
	$B$\> The diagonalization of $b$\\
	$L$, $H$\> The Laplacian matrix of $\cal G$ and $H=L+B$\\
	$\eta_i(H)$\> Eigenvalue of $H$\\
	$\eta_m$, $\eta_M$\> The maximum and minimum values of $\eta_i(H)$\\
	$\xi_i$, $\nu_i$\> Proportional and integral terms \\
    $k$, $h$\> Control gains in $u_i$\\
    $\xi$, $\nu$, $x$\> Some useful stack vectors\\
    $\Phi$\> A matrix for the base system\\
    $\Phi_0$\> A matrix for the reset mechanism\\
    $t_l^i$. $T_i$\> Reset instant and its set\\
	$\mu_i$\> Eigenvalue of $\Phi$\\
	$e_i$, $e$\> The state error and its stack vector\\
	$V$, $V_d$\> Two candidate Lyapunov functions\\
	${\cal L}_u V$\> The set-valued derivative\\
	$\partial V$\> The generalized gradient\\
	$r_i$\> The ratio of $\xi_i$ to $\nu_i$ in real-time\\
	$\Delta$\> A constant dwell-time\\
\end{tabbing}

\section{Introduction}

\IEEEPARstart{A}{s} more renewable energy is integrated to communities in the form of distributed power generation, its randomness and intermittence bring that power provided by microgrid (MG) suffers low quality, thus endangering the utility grid (UG) \cite{wangApplicationEnergyStorage2022, solyaliComprehensiveStateoftheartReview2022, ghanjatiOptimalSizingEnergy2022}. In addition, the power demand of various loads is gradually increasing, and the power shortage gradually appears in the grid-connected MG during the peak hours of power consumption \cite{ghanjatiOptimalSizingEnergy2022, yangModellingOptimalEnergy2022, dahiruComprehensiveReviewDemand2023}. To alleviate the current difficulties, efforts are made by both the supply side and the demand side to this end. Therein, to guide users to use electricity reasonably, the time-phased electricity price is adopted by electric company (EC) \cite{wangMultiagentbasedCollaborativeRegulation2022, hossainlipuReviewControllersOptimizations2022}. For MG shareholders, fortunately, battery energy storage systems (BESSs) are introduced, which are conducive to make up for the shortcomings of renewable energy, thus improving the power supply quality of MG and also greatly easing the power tension during peak load \cite{rouholaminiReviewModelingManagement2022,caleroReviewModelingApplications2022}. Besides, the immediate relationship between electricity production and consumption has been severed by BESSs, which realizes the transfer of power along the time dimension \cite{rouholaminiReviewModelingManagement2022, caleroReviewModelingApplications2022, ulhassanComprehensiveReviewBattery2022}.
\par Although the introduction of BESS has improved the current predicament, the problems that follow have gradually emerged. Both in the charging and discharging modes, a certain amount of power is consumed by the internal impedance of the battery cell \cite{hossainlipuReviewControllersOptimizations2022,yuFrequencySynchronizationPower2021b}. Due to long-term operation, the capacity of battery cells will be reduced due to aging \cite{hossainlipuReviewControllersOptimizations2022}. All these leads to the increase in the operating costs of MG. Therefore, for grid-connected BESSs, it is necessary to design an economic dispatch (ED) algorithm for MG shareholder to reduce costs. In addition, there is electricity trading between the shareholder of MGs and EC \cite{forero-quinteroProfitabilityAnalysisDemandside2022,tanExtensionsLocationalMarginal2022}. That is, excess electricity is sold to EC. Or, insufficient electricity is purchased from EC.  
\par Currently, many excellent solutions to the ED problem of MG are given by scholars, both in centralized or distributed manners. With regard to the ED problem, many excellent results are achieved by using centralized methods, such as $\lambda$ iteration \cite{wangDistributedOptimalPower2021}, gradient search \cite{liNewDistributedEnergy2021,hassanImprovedMantaRay2021} and some intelligent search algorithms \cite{wangMultiagentbasedCollaborativeRegulation2022}. However, the centralized method has been gradually replaced by the distributed method for the limitations of high communication cost, heavy computing burden, lack of privacy protection, etc.
\par Recently, the distributed ED (DED) scheme based on multi-agent systems (MASs) has been favored by scholars. In these ED problems, some optimizable objectives, such as, generation cost \cite{ullahComputationallyEfficientConsensusBased2021, saeidiniaAutonomousControlDC2023, xuConsensusActivePower2021, sahooLocalizedEventDrivenResilient2021, zaeryNovelFullyDistributed2021, alviNovelIncrementalCost2022,songCostBasedAdaptiveDroop2021, wangDisEHPPCEnablingHeterogeneous2022, liDistributedControlStrategy2021, pengDistributedPeriodicEventTriggered2022, chenDistributedEconomicDispatch2021}, transmission loss \cite{ullahComputationallyEfficientConsensusBased2021, saeidiniaAutonomousControlDC2023}, etc, are considered singly or in combination. Besides, whether these cost functions are single or combined, they are strongly convex of output power in these distributed schemes. In these distributed schemes, authors in \cite{ullahComputationallyEfficientConsensusBased2021, wangDisEHPPCEnablingHeterogeneous2022,liDistributedControlStrategy2021, pengDistributedPeriodicEventTriggered2022, chenDistributedEconomicDispatch2021} address the ED problems on a large time scale. These scholars provide some good solutions for the tertiary control of MG. For the rest, another group of researchers in \cite{ saeidiniaAutonomousControlDC2023, xuConsensusActivePower2021, sahooLocalizedEventDrivenResilient2021, zaeryNovelFullyDistributed2021, alviNovelIncrementalCost2022,songCostBasedAdaptiveDroop2021} seem to show solicitude for the schemes on a small time scale. This is because they try to construct economic droop control with optimal power to achieve the restoration of frequency or voltage both in AC and DC MGs, i.e., secondary control. Without exception, these schemes are methods based on MC \cite{tanExtensionsLocationalMarginal2022, spanglerPowerGenerationOperation2014, binettiDistributedConsensusBasedEconomic2014, babazadeh-dizajiDistributedHierarchicalControl2020} consensus to address the ED problem of MG, which can be effectively transplanted to BESSs. Most of these schemes are proportional controllers (PCs). Occasionally, finite/fixed time control schemes \cite{xuConsensusActivePower2021, zaeryNovelFullyDistributed2021} are also designed to accelerate system convergence.
\par In fact, some scholars have made a preliminary exploration on the ED problem in BESSs, and unexpectedly designed a distributed scheme based on MC. In a MG containing distributed generators and BESSs, a PC based on a distributed consensus protocol is designed in \cite{yuFrequencySynchronizationPower2021b} to reach consensus on MC so that each BESS behaves the optimal charge/discharge loss power while frequency is restored. Authors in \cite{chenDistributedCooperativeControl2021} design a distributed control algorithm to reach MCs consensus, and then realize the optimal ED of BESSs, in which the power consumed by the internal impedance of a battery is considered in the objective function. Also, charging/discharging efficiency is optimized in \cite{jinManageDistributedEnergy2022, zhaoDifferentialPrivacyEnergy2022} on a large time scale. In addition, battery capacity degradation is considered in the objective function in \cite{sitchEconomicControlHybridElectric2022} on a large time scale. 
\par The above authors provide excellent distributed solutions for ED problems in MGs and BESSs on a large or small time scale. Unfortunately, these control schemes are all PCs. A PC is effective but not efficient for ED problems, which encounters slow convergence speed and low control accuracy. Besides, great difficulties in physical implementation and chattering on MCs and control input are brought by the finite/fixed time controller with fractional order terms. Therefore, this paper focuses on the reset controller \cite{banosResetControlSystems2011, zhaoResetControlConsensus2021, chengResetControlLeaderfollowing2022, huResetControlConsensus2022,  mengResetControlSynchronization2019, zhaoResetControlConsensus2021a} based on the proportional-integral (PI) controller, and takes the composite objective of capacity degradation, charge/discharge loss and the cost of purchasing electricity in BESS on a small time scale as the research object, trying to design an ED algorithm with MC fast response and convergence speed and high control accuracy. Therefore, BESSs in a resistive network can achieve the optimal discharge power to reduce operational expenses. The specific work is listed as follow:
\begin{itemize}
	\item A distributed PI controller with reset mechanism is designed to reach MC consensus, which can improve the dynamic and steady-state performance and contribute to solve the multi-objective ED problem in BESSs.
	\item Both MC consensus and reset regularity are well analyzed, and parameter conditions are given. Regularity ensures that the reset mechanism is enabled.
	\item Zeno behavior and progressiveness of the controller are analyzed. A Zeno-free reset mechanism is given. The MC consensus under a PC and the finite-time controller are compared with the designed scheme to illustrate the progressiveness.
\end{itemize}
\par Next, the multi-objective ED problem in BESS is analyzed in section II. And the theoretical analysis of the designed scheme is arranged in section III. Several simulation cases and related analysis are arranged in section IV. Finally, the conclusion is draw in section V. 
\section{Preliminaries}
\par For a resistive network, its operator needs to pay a fee for the absorbed power or charge the power company for the power released into the grid. Due to the internal impedance, each battery encounters different degrees of power loss. With the passage of time, the capacity of each battery will deteriorate to varying degrees. In this section, a convex expenditure model considering multiple objectives is constructed, which contains an income model, a inner power loss model in charge/discharge mode and a energy storage capacity deterioration model. Based on the MC, the optimal ED on a small time scale is analyzed, and the control objectives are given.
\subsection{BESS and its expenditure model}
For a set multiple BESSs containing $n$ batteries integrated to smart grids, the residual power is fed back to UG so that the shareholders of BESSs obtain the profits provided by EC. Or, power is absorbed from UG and partially stored. In this case, the shareholders should pay for this energy. Therefore, the income model \cite{hossainlipuReviewControllersOptimizations2022, chenDistributedEconomicDispatch2021, zhaoDifferentialPrivacyEnergy2022} of BESSs is as follow
\begin{equation}
	C^P=\rho P_e,
\end{equation}
where $P_e$ is the exchange power between BESSs and UG, $\rho$ is the agreed electricity price between the shareholders and EC, $P_e>0$ ($P_e<0$) means that BESSs absorb (release) power from (to) UG. The inner power loss ratio of charge/discharge model \cite{yuFrequencySynchronizationPower2021b, chenDistributedCooperativeControl2021, jinManageDistributedEnergy2022, zhaoDifferentialPrivacyEnergy2022} of a BESS is often a linear function of output power on a small time scale, i.e.,
\begin{equation}
	\pi_i=\alpha_iP_i+\beta_i,
\end{equation}
where $\alpha_i$ and $\beta_i$ are two constant coefficients of the inner loss ratio $i$. Energy storage capacity deterioration cost \cite{hossainlipuReviewControllersOptimizations2022} is expressed as
\begin{equation}
	f_{Li}=m_i\pi_iP_i,
\end{equation}
where $m_i$ is called as the capacity deterioration coefficient, $P_i$ is the output power of BESS $i$. Considering income, loss and capacity deterioration cost, the following multi-objective operating expenditure function is constructed in (\ref{4}), which is subjected to the supply and demand balance,
\begin{equation}
	\begin{array}{l}
		{F_c} = \sum\limits_{i = 1}^n {({\eta _i}{P_i} + {f_{Li}})}+ {C^P}\\
		\quad\,\, = \sum\limits_{i = 1}^n {(1 + {m_i})({\alpha _i}P_i^2 + {\beta _i}{P_i})}+\rho (t){P_e}\\
		{\rm{         }}s.t.\quad \sum\limits_{i = 1}^n {{P_i}+{P_e}-D}  = 0,
	\end{array}\label{4},
\end{equation}
where $D$ is the total load, the expenditure function (\ref{4}) involves BESS and the UG operated by EC. It is necessary to develop a reasonable output power strategy to reduce cost. 
\subsection{Analysis for optimal ED}
\par Since the expenditure function (\ref{4}) is convex, the optimal condition may exist for BESSs. To this end, construct a Lagrangian function $L_{ac}$ for (\ref{4}),
\begin{equation}
	L_{ac}=F_c-\gamma(\sum\limits_{i = 1}^n {{P_i}+{P_e}-D}  = 0),
\end{equation}
where $\gamma$ is a Lagrange multiplier. Thus, we can get the KKT conditions for BESSs as follow
\begin{subequations}
	\begin{equation}
		(1+m_i)(2\alpha_iP_i+\beta_i)-\gamma=0 
		\label{6a},
	\end{equation}
	\begin{equation}
		\rho-\gamma=0
		\label{6b},
	\end{equation}
\end{subequations}
where $\lambda_i=(1+m_i)(2\alpha_iP_i+\beta_i)$ is called MC, (\ref{6a}) and (\ref{6b}) is the optimal ED condition for BESS $i$. From (\ref{6a}) and (\ref{6b}), it can conclude that the total MG expenditure is the lowest if the MC of each battery reaches electricity price.
\subsection{Control objective}
A simplified discharge control structure with an ED algorithm for a grid connected BESS is shown in Fig. \ref{figei}, which contains a battery unit, a DC-DC chopper, an DED controller, and a current calculation unit. To design an ED scheme, MC of each BESS is given by the following dynamic on a small time scale,
\begin{equation}
	\label{7}
	\dot\lambda_i=u_i.
\end{equation}
Denote two vectors as $\lambda=[\lambda_1, \lambda_2,\cdots, \lambda_n]^T$, $u=[u_1, u_2,\cdots, u_n]^T$.
\begin{figure}
	\centering
	\includegraphics[width=7cm]{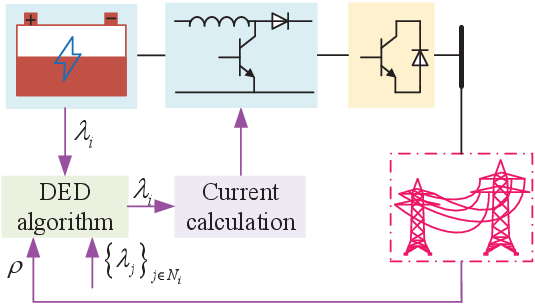} \caption{MC consensus based ED framework for BESSs.\label{figei}}
\end{figure}
\par For a multiple BESSs in a resistive network, ensuring its economic operation has become the general goal and basic principle of controller design. In this paper, we pay more attention to reduce the total net expenditure. Thus, we are going to design a DED controller so that the total expenditure function of the grid-connected BESSs achieves the optimal solution if MCs reach the electricity price. That is, mathematically,
	$$\lim_{t \to +\infty}{\lambda_i}=\rho,$$
where $\rho$ denotes the time-phased electricity price. Generally, EC often signs time-phased electricity price contracts with users.
\section{DED algorithm with reset mechanism}
\par To obtain the optimal power flow and reduce the expenditure of operating the BESSs, most of the common MC-based distributed schemes are PC or finite-time protocols, which exhibit low control accuracy or chattering, respectively. Thus, a distributed PI algorithm with reset mechanism based on MASs is designed, as shown in Fig. \ref{figr}, to drive MCs to reach the time-phased electricity price more quickly, accurately and without chattering. Stability, regularity, progressiveness and Zeno behavior are well analyzed.
\begin{figure}
	\centering
	\includegraphics[width=8cm]{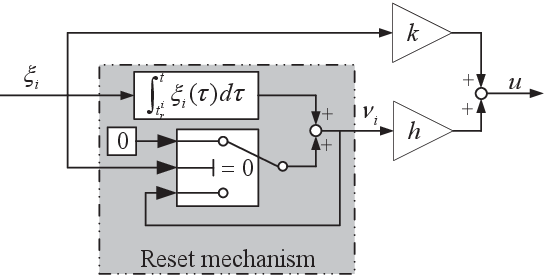} \caption{The designed DED algorithm with reset mechanism.\label{figr}}
\end{figure}
\subsection{Graph theory and some related conclusions}
\par In the distributed control scheme of BESSs based on MASs, each BESS is governed by an agent. Each agent communicates with others over a communication graph $\cal G$, where each agent and communication link are denoted as a node and an edge respectively. Denote ${\cal V}=\left\{v_1,\cdots, v_n\right\}$ as the node set. It is assumed that $(v_i, v_j)$ implies that agent $i$ can access agent $j$, which is so-called a neighbor of agent $i$. Define the edge set as ${\cal E}=\left\{(v_i, v_j)|v_i,v_j\in {\cal V} \right\}$ and neighbor set of agent $i$ as $N_i=\left\{v_j|(v_i, v_j)\in \cal E\right\}$. Denote $d_{ii}$ as the number of neighbors of agent $i$ and a diagonal matrix ${\cal D}=diag\left\{d_{11}, \cdots, d_{nn}\right\}$. An adjacency matrix ${\cal A}=[a_{ij}]$ signifies that $a_{ij}=1$ if $j\in N_i$ and $a_{ij}=0$ otherwise. Thus, the Laplace matrix  corresponding to the graph $\cal G$ can be expressed as $L={\cal D}-{\cal A}$, which plays an important role in controller design. 
A common assumption is given in Assumption \ref{assumption 1}.
\begin{assumption}
	\label{assumption 1}
	\cite{linNecessarySufficientGraphical2005} 
	In this paper, the communication graph $\cal G$ is undirected and connected.\label{assum 1}
\end{assumption}
\par Besides, for a MASs with a leader, an adjacency vector $b=[b_{11}, \cdots, b_{nn}]$ needs to be defined, where $b_{ii}$ implies that agent $i$ can access the leader. Denote a diagonal matrix $B$ as $B=diag\left\{b_{11}, \cdots, b_{nn}\right\}$. Thus, a common lemma without proof is given as follows.
\begin{lemma}
	\cite{niLeaderfollowingConsensusMultiagent2010} 
	Denote $H=L+B$, and $H$ is positive define. The eigenvalues of $H$ are ordered as $0<\eta_m=\eta_1(H)\le \eta_2(H)\le \cdots \le \eta_n(H)=\eta_M$.
\end{lemma}
\subsection{DED algorithm with reset mechanism}
In the grid-connected mode, BESSs can absorb/release power from/to UG. According to the analysis in subsection C of section II, the total expenditure of BESSs is optimal when MC of each BESS reaches the time-phased electricity price. 
With the support of the reset control, the basic performance limitations of linear control system design can be overcome, i.e., the contradiction between overshoot and response speed. Here, we will give the DED algorithm in the leader-following manner. To improve control accuracy, avoid overshoot and speed up convergence, reset mechanism is introduced.
\par Define the state difference in below
\begin{equation}
	\label{8}
	\xi_i=\sum_{i=1}^{n}a_{ij}(\lambda_i-\lambda_j)+b_{ii}(\lambda_i-\rho(t))
\end{equation}
Then, the ED scheme under the guidance of electricity price is as follows
\begin{subequations}
	\begin{equation}
		\label{9a}
		\dot \nu _i =\xi _i,
	\end{equation}
	\begin{equation}
		\label{9b}
		u_i=-k\xi_i-h\nu_i,
	\end{equation}
\end{subequations}
where $k,h>0$ are the control coefficients, the reset mechanism based on a so-called Clegg integrator \cite{cleggNonlinearIntegratorServomechanisms1958} is designed below,
\begin{equation}
	\left\{\begin{array}{l}
		\dot \nu _i =\xi _i,\quad \,\,\,\,if\,\xi_i \ne 0\\
		\nu _i ({t^ + }) = 0,\,if\, \xi _i = 0
	\end{array} \right.\label{10},
\end{equation}
where the initial value of $\nu_i$ is 0. It is easy to observed that $\nu_i$ has a same sign as $\xi_i$. Denote  $\nu=col\left\{\nu_1,\nu_2,\cdots,\nu_n\right\}$ and $\xi=col\left\{\xi_1,\xi_2,\cdots,\xi_n\right\}$. Thus, the following closed loop dynamic system is constructed
\begin{subequations}
	\begin{equation}
		\dot x=\Phi x, \label{11a}
	\end{equation}
    \begin{equation}
	    x(t^+)=\Phi_0 x, \label{11b}
    \end{equation}
\end{subequations}
where $x=[\xi^T\quad \nu^T]^T$, $\Phi=\left[\begin{matrix}
	-kH & -hH\\
	I & 0
\end{matrix}\right]$, $\Phi_0=\left[\begin{matrix}
I & 0\\
0 & 0
\end{matrix}\right]$.
\par To facilitate subsequent analysis, regularity definition and some useful lemmas without proof are given as follows.
\begin{definition}
	(\ref{11a}) is referred to as the base system \cite{banosResetControlSystems2011}, and (\ref{10}) and its equivalent form (\ref{11b}) are called as reset mechanism, respectively. The base system with reset mechanism is regular if the base system is stable and the reset mechanism is ensured to be enable. That is, define a set $T_i=\left\{t^i_l\right\}_{l=1}^{\infty}$ and $t^i_l$ is called the reset instant for agent $i$, $i\in \cal V$. If there exists at least one agent $i\in \cal V$ with $T_i\ne \emptyset$ and $t^i_l<\infty$ \cite{mengResetControlSynchronization2019} and all eigenvalues of $\Phi$ are with negative real parts, the DED scheme, i.e., (\ref{9a}) and (\ref{9b}) with (\ref{10}), is regular. \label{Definition 1}
\end{definition}
\begin{lemma}
	As for $\Phi$ of the base system, only if the dominant ones are complex, there exist $i\in [1,2,\cdots, n]$ so that $\nu_i$ is reset to 0 at least once \cite{banosResetControlSystems2011}. That is, all nonzero eigenvalues are complex \cite{mengResetControlSynchronization2019}. \label{Lemma 3}
\end{lemma}
\par With the aforementioned lemmas in hand, we give the following two theorems on the designed DED scheme, where one gives the regularity condition and the other analyzes the consensus condition.
\begin{theorem}
	\label{Theorem 1}
	Under Assumption \ref{assum 1}, the ED scheme (\ref{9a}) and (\ref{9b}) with the reset controller (\ref{10}) is regular if $k>0$ and
	\begin{equation}
		4h>\eta_n(H)k^2 \label{12}.
	\end{equation}
\end{theorem}
\begin{proof}
	The characteristic polynomial of $\Phi$ is
	$$\lvert \mu I-\Phi\rvert=\prod_{i=1}^n {\lvert\mu_i^2+k\eta_i(H)\mu_i+h\mu_i(H)\rvert},$$
	which gives $n$ nonzero eigenvalues with negative real parts as follow
	$$\mu_i=\frac{-k\eta_i(H)\pm \sqrt{k^2(\eta_i(H))^2-4h(\eta_i(H))}}{2}.$$
	Thus, according to Lemma \ref{Lemma 3}, only if $4h>\eta_n(H)k^2$, all eigenvalues of $\Phi$ are complex. Thus, zero crossings of $\xi$ occur at least once. According to Definition \ref{Definition 1}, the designed DED scheme (\ref{9a}) and (\ref{9b}) with (\ref{10}) is regular, i.e., Theorem \ref{Theorem 1} holds.
\end{proof}
\begin{theorem}
	\label{Theorem 2}
	Under Assumption \ref{assum 1}, the DED scheme (\ref{9a}) and (\ref{9b}) with the reset controller (\ref{10}) can achieve consensus with at least one reset if $k>0$ and $4h>\eta_n(H)k^2$.
\end{theorem}
\begin{proof}
	It can be observed that $\nu$ is discontinuous in time. Hence, the Filippov solution to (\ref{7}) under (\ref{9a}) and (\ref{9b}) with (\ref{10}) is concerned here.
	\par Define the state error $e_i=\lambda_i-\rho$ and $e=[e_1, \cdots, e_n]^T$. Select a candidate Lyapunov function
	\begin{equation}
		\label{13}
		V=\mathop{max}\limits_{i\in N_i}(e^2_i).
	\end{equation}
	It can be easily concluded that $V\ge 0$ with equality if and only if $e_1=\cdots=e_n=0$. 
	\par Define a set-valued Lie derivative as 
	$${\cal L}_u V=\partial V^T u,$$
	where $\partial V$ is the generalized gradient of $V$ with respective to $e$ and 
	$$\partial V=[\zeta_1 e_1,\cdots,\zeta_n e_n]^T,$$
	where $\zeta=[\zeta_1,\cdots,\zeta_n]^T$, $\zeta_i=2$ if $i=\mathop{argmax}\limits_{j\in [1,\cdots,n]}(e^2_j)$ and $\zeta_i=0$ otherwise. Thus, 
	$${\cal L}_u V=\sum_{i=1}^{n}(\zeta_i e_i u_i).$$
	Denote ${\cal H}_{ij}$ as ${\cal H}=[{\cal H}_{ij}]_{n\times n}$. Hence, $\xi_i=\sum_{j=1}^{n} {\cal H}_{ij} e_j$ and $\nu_i=r_i\xi_i$. Thus, one can get 
	\begin{equation}
		\label{14}
		{\cal L}_u V=-\sum_{i=1}^{n}((k+hr_i)\zeta_i\sum_{j=1}^{n}{\cal H}_{ij} e_i e_j).
	\end{equation}
	Note that
	\begin{equation}
		\label{15}
		\begin{array}{l}
			\zeta_i\sum\limits_{j = 1}^n {{{\cal H}_{ij}}} {e_i}{e_j} = {{\cal H}_{ii}}e_i^2 + \sum\limits_{j = 1,j \ne i}^n {{{\cal H}_{ij}}} {e_i}{e_j}\\
			\qquad \qquad \qquad \ge \zeta_i {{\cal H}_{ii}}e_i^2+\zeta_i\sum\limits_{j = 1,j \ne i}^n {\cal H}_{ij}(\frac{1}{2}(e_i^2 + e_j^2))\\
			\qquad \qquad \qquad \ge -\frac{1}{2}\zeta_i(\sum\limits_{j = 1,j\ne i}^n{{\cal H}_{ij}(e_i^2-e_j^2)})\ge 0.
		\end{array}.
	\end{equation}
	Thus,
	\begin{equation}
		\label{16}
		{\cal L}_u V\le\sum_{i=1}^{n}{({\frac{1}{2}}(k+hr_i)\zeta_i(\sum\limits_{j = 1,j\ne i}^n{{\cal H}_{ij}(e_i^2-e_j^2)}))}\le 0,
	\end{equation}
	only if $k,h>0$. According to LaSalle Invariance Principle, all components of $e$ converge to the largest invariant set
	\begin{equation}
		\label{17}
		\{e|{\cal L}_u V=0\},
	\end{equation} 
	i.e., $e_1=\cdots=e_n=0$ as $t\rightarrow \infty$. As a result, $\lambda_1=\cdots=\lambda_n=\rho$ as $t\rightarrow \infty$. With $k>0$ and $4h>\eta_n(H)k^2$, consensus and regularity can be guaranteed. That is, Theorem \ref{Theorem 3} is proved.
\end{proof}
\begin{remark}
	\label{remark 1}
	Under (\ref{9a}) and (\ref{9b}), MCs convergence can be accelerated by using a large $k$. However, it is unwise to use excessive $k$ because it will cause an excessive overshoot. Under the designed control scheme, the integral term is reset to 0 when the proportional term encounters zero crossing so that the signs of the two are aligned. Consequently, the overshoot is avoided. Thus, under this control scheme, a slightly larger $k$ is allowed to speed up MC convergence. Besides, it can be concluded from Theorem \ref{Theorem 2} that the reset mechanism behaves no threat to consensus only if $k,h>0$, and parameter condition for the regular DED scheme (\ref{9a}) and (\ref{9b}) with (\ref{10}) depends on the communication graph.
\end{remark}
\subsection{Further discussion on the designed schemes}
Under the DED scheme (\ref{9a}) and (\ref{9b}) with the reset mechanism (\ref{10}), each integral term is reset to zero when its corresponding proportional item is about to change its sign. Thus, this reset control system possesses a better transient performance, such as a shorter settling time and no overshoot, than both a PC $\dot \lambda_i=-k\xi_i$ and a PI protocol (\ref{9a}) and (\ref{9b}).
\par Compared to a PC, an additional dynamic on $\nu$ is introduced in a reset controller, resulting in the order growth of the closed-loop system. Actually, (\ref{9a}) with (\ref{10}) can be reduced as follow
\begin{equation}
	\left\{\begin{array}{l}
		\nu _i = r_i\xi _i,\\
		r_i=\frac{\int_{t_l^i}^{t}\xi_i(\tau) d\tau}{\xi _i}.
	\end{array} \right.\label{18}
\end{equation}
Thus, the dynamic consensus protocol (\ref{9a}) and (\ref{9b}) with (\ref{10}) are transformed into a PC with dynamic gain $(k+hr_i)$. 
This design can bring better transient performance than a static protocol with a fixed gain. Further, we give another simplified method. If the reset interval is small enough, (\ref{18}) is approximately equivalent to $r_i=t-t_l^i$. 
\begin{remark}
	\label{remark 3}
	It is precisely because $\xi_i$ and $\nu_i$ have the same sign that $r_i$ is always greater than 0, i.e., $k<(k+hr_i)$. Thus, compared to a PC with a fixed gain $k$ \cite{yuFrequencySynchronizationPower2021b, chenDistributedCooperativeControl2021}, $e$ will converge faster to 0 under the designed DED scheme. Besides, from (\ref{16}), one can get the set $\left\{e|V\le V(e(0))\right\}$ is compact and strongly invariant under the designed DED scheme, which is, however, not a invariant set for a non-reset PI control system. This is because the sign of $\nu_i$ is not always aligned with $\xi_i$. Hence, $\Vert V \Vert_\infty\le \Vert V_d \Vert_\infty$, where $V_d=\mathop{max}\limits_{i\in N_i}(e^2_i)$ is a candidate Lyapunov function for a non-reset PI control system. 
\end{remark}
\begin{remark}
	\label{remark 6}
	For a DED scheme, the finite-time and fixed-time methods are a common scheme \cite{xuConsensusActivePower2021,zaeryNovelFullyDistributed2021}. However, chattering often occurs in the simulation of these controllers by using symbolic functions. Besides, it is difficult to deal with fractional item in practical application. The reset controller is an improved PI controller, which resets the integral term when the proportional item is 0. This not only avoids the chattering phenomenon, but also reduces the control input and accelerates the convergence of the system. At the same time, the  reset controller is relatively simple and easy to realize physically.
\end{remark} 
\par To prevent Zeno behavior of the designed algorithm, a constant dwell-time $\Delta$ is given here. And a theorem on the stability of (\ref{11a}) and (\ref{11b}) without Zeno behavior is given in Theorem \ref{Theorem 3}.
\begin{theorem}
	\label{Theorem 3}
	For the designed DED scheme (\ref{9a}) and (\ref{9b}) with (\ref{10}), if the reset mechanism is modified as (\ref{19})
	\begin{equation}
		\label{19}
		\nu_i(t^+)=0,\;if\;\xi_i=0\;and\;t-t_l^i\ge l\Delta,
	\end{equation}
	where $\Delta$ is an arbitrary positive constant, the designed dynamic DED scheme with reset mechanism possesses the asymptotical stability without Zeno behavior.
\end{theorem}
\begin{proof}
	From Theorem \ref{Theorem 1} and Theorem \ref{Theorem 2}, the solution of the base system (\ref{11a}) is $x=exp({\Phi_0(t-t_l)})x(t_l^+)$, where $t_l<t\le t_{l+1}$ and $\Phi$ in (\ref{11a}) is Hurwitz matrix. Thus, there exist an invertible matrix $W$ so that $\Phi=W^{-1}JW$, where $J$ is a Jordan matrix and each Jordan block is Hurwitz. Then 
	$$\begin{array}{l}
		x({t_{l + 1}}) = {W^{ - 1}}exp(J({t_{l + 1}} - {t_l}))Wx(t_l^ + )\\
		\qquad \quad \,\,\,= {W^{ - 1}}exp(kJ\Delta )Wx(t_l^ + ).
	\end{array}$$
	Take 1-norm on both sides of the above equation, 
	$${\Vert x(t_{l+1})\Vert}_1 \le {\Vert W^{-1}\Vert}_1 {\Vert exp({lJ\Delta})\Vert}_1 {\Vert W \Vert}_1 {\Vert x(t_l^+)\Vert}_1.$$
	There always exists a sufficiently large constant $k^*$ so that if $k\ge k^*$, ${\Vert W^{-1}\Vert}_1 {\Vert exp({kJ\Delta})\Vert}_1 {\Vert W \Vert}_1\le 1$. Thus, one can get
	$${\Vert x(t_{l+1})\Vert}_1<{\Vert x(t_{l}^+)\Vert}_1.$$
	\par At reset instant, there is
	$$x(t_{l+1}^+)=\Phi_0 x(t_{l+1}).$$
    Also, take 1-norm on both sides of the above equation,
    $${\Vert x(t_{l+1}^+)\Vert}_1\le {\Vert x(t_{l+1})\Vert}_1.$$
     Thus, $
    \lim\limits_{t\to+\infty} {\Vert x\Vert}_1=0$. Combined with the fact that ${\Vert x\Vert}_2\le {\Vert x\Vert}_1$, one can get
    $$\lim\limits_{t\to+\infty}{\Vert x\Vert}_2= 0.$$
    So far, Theorem \ref{Theorem 3} has been proved.
\end{proof}
\begin{remark}
	\label{remark 5}
	From the above proof, we can see that for the ED scheme with finite number of reset instants, the stability of MCs mainly depend on the dynamics of non-reset times. However, excessive a dwell time $\Delta$ may lead to deterioration of system dynamic performance. Because the reset times are artificially reduced, the integral function in PI controller will be strengthened after zero crossing.
\end{remark}
\section{Case study}
Consider a MG containing $4$ BESSs and an ER governed by a MAS. Single line diagram of the 4 BESSs with an ER and their equipped communication topology of MASs are shown in Fig.~\ref{fig3}. Therein, electricity price as the leader is provided by an ER, and only agent 1 can obtain the leader's information. Accordingly, $\eta_n(H)=4.7913$. 7 cases are arranged here to verify the analysis in Section III.
\begin{figure}
	\centering
	\includegraphics[width=8cm]{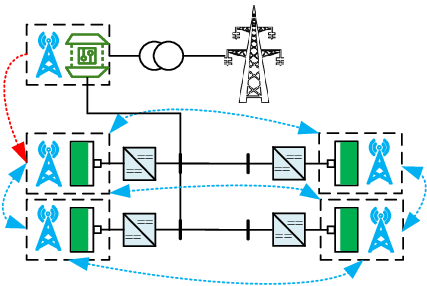} \caption{Single line diagram of a DC microgrid containing $4$ BESSs and an ER with a communication graph.\label{fig3}}
\end{figure}
\par In Case 1, the effectiveness of the proposed distributed reset control scheme, i.e., Theorem \ref{Theorem 1} and \ref{Theorem 2}, is verified under $k=1$ and $h=2$. At this point, the inequality (\ref{12}) holds.
\par In Case 2, the progressiveness of the DED algorithm with reset mechanism is verified under $k=1$ and $h=2$, as shown in Remark \ref{remark 3} and \ref{remark 6}.
\par In Case 3, the influence of parameters $k$ and $h$ on MCs is discussed.
\par In Case 4, the plug-and-play test on the designed algorithm is performed under the conditions $k=1$ and $h=2$.
\par In Case 5, the performance of the designed algorithm is tested in the time-phased electricity price environment.
\par In Case 6, the DED algorithm under a Zone-free condition with reset mechanism is verified with $\Delta=1$ and $\Delta=1.5$, as shown in Theorem \ref{Theorem 3} and Remark \ref{remark 5}.
\par In Case 7, the number of BESSs and the agents that manage them is increased to 12 such that the scalability applicable to a large-scale power system can be verified.
\subsection{Case 1: Validation of the Proposed DED algorithm}
\begin{figure}
	\centering
	\includegraphics[width=8cm]{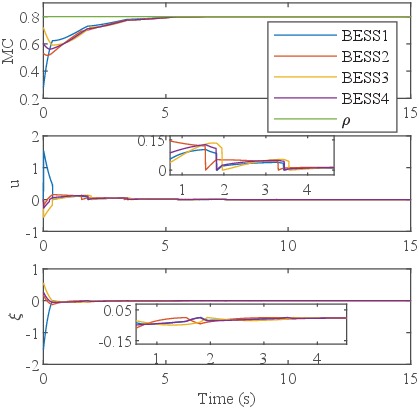} \caption{The evolution of MCs, $u_i$ and $\xi_i$ under the designed DED algorithm in Case 1.\label{figc1}}
\end{figure}
\begin{figure}
	\centering
	\includegraphics[width=8cm]{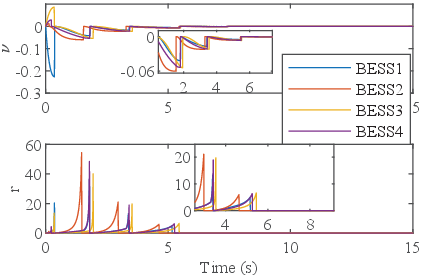} \caption{The evolution of $\nu_i$ and $r_i$ under the designed DED algorithm in Case 1.\label{figc1c}}
\end{figure}
\par Fig. \ref{figc1} and \ref{figc1c} are the simulation results under the designed DED algorithm, where contains the evolution of MCs, control input vector $u$, the state difference vector $\xi$, integral term $\nu$ and the dynamic gain vector $r$ in (\ref{18}). In this scenario, constant electricity price $\rho=0.8$ is adopted.
\par As can be seen from Fig. \ref{figc1}, MCs can converge to constant electricity price $\rho$. Meanwhile, control input $u$ and state difference $\xi$ converge to 0. From Fig. \ref{figc1c}, integral term $\nu$ and dynamic gain $r$ are reset several times when the state difference occurs zero crossings and finally converges to 0. Besides, the dynamic gains are always guaranteed to be non negative, which also indicates that the integral and proportional terms are always sign aligned. Therefore, the algorithm designed is effective.
\subsection{Case 2: Discussion on the progressiveness}
\begin{figure}
	\centering
	\includegraphics[width=8cm]{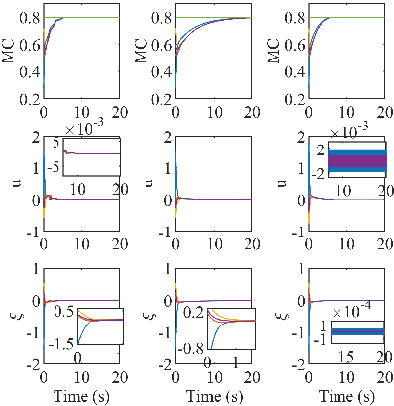} \caption{Simulation results under the designed controller, a PC and a finite-time controller, which are listed in columns from left to right in Case 2.\label{figc2}}
\end{figure}
\par To illustrate progressiveness, comparison of simulation results of three DED algorithms, i.e., the designed controller, a PC and a finite-time controller, is shown in Fig. \ref{figc2}.
\par It is not difficult to see that MCs is driven to the electric price $\rho$ faster and without overshoot by the designed algorithm, comparing the subgraphs of the first two columns. Although the convergence rate of MCs can be accelerated under a finite-time controller, the chattering of control input and MCs caused by the fractional integral term is unavoidable, as shown in subfigures of the third column. In addition, at present, it is physically difficult to achieve fractional integration in the project. The controller designed in this paper is developed from PI controller, which is widely used in industry. Therefore, it is reasonable to believe that the controller designed is more convenient for engineering application.
\subsection{Case 3: Algorithm Test on different parameters}
\begin{figure}
	\centering
	\includegraphics[width=8cm]{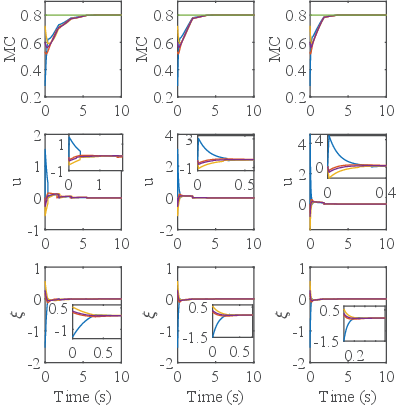} \caption{Simulation results on MCs, $u_i$ and $\xi_i$ with $k$=1, 5, 10 and $h=2$, which are listed in columns from left to right in Case 3.\label{fig3k} }
\end{figure}
\begin{figure}
	\centering
	\includegraphics[width=8cm]{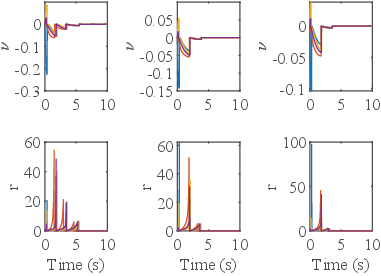} \caption{Simulation results on $\nu_i$ and $r_i$ with $k$=1, 5, 10 and $h=2$, which are listed in columns from left to right in Case 3.\label{fig3ck}}
\end{figure}
\begin{figure}
	\centering
	\includegraphics[width=8cm]{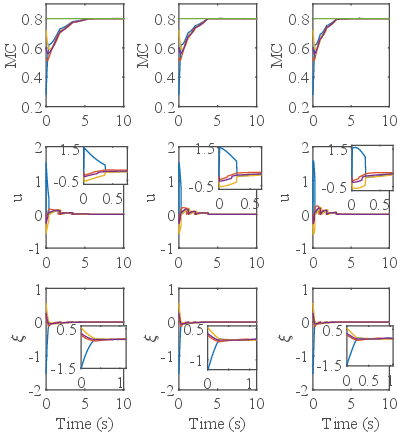} \caption{Simulation results on MCs, $u_i$ and $\xi_i$ with $k=1$, $h$=2, 6 and 10, which are listed in columns from left to right in Case 3.\label{fig31h}}
\end{figure}
\begin{figure}
	\centering
	\includegraphics[width=8cm]{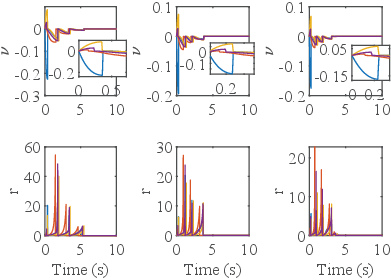} \caption{Simulation results on $\nu_i$ and $r_i$ with $k=1$, $h$=2, 6 and 10, which are listed in columns from left to right in Case 3.\label{fig31ch}}
\end{figure}
\par In Remark \ref{remark 1}, the impact of parameters on controller performance has been analyzed. Here, we use the trial and error method to visually illustrate these results. Fig. \ref{fig3k}-\ref{fig31ch} are the simulation results under several sets of parameters. In Fig. \ref{fig3k} and \ref{fig3ck}, the subgraphs in each column from left to right are the simulation results when $k$=1, 5, 10 and $h=2$. In Fig. \ref{fig31h} and \ref{fig31ch}, the subgraphs in each column from left to right are the simulation results when $k=1$, $h$=2, 6 and 10.
\par As $k$ and $h$ increase in Fig. \ref{fig3k} and \ref{fig31h}, MCs converge to the electricity price $\rho$ faster. At the same time, the control input $u_i$ and state difference $\xi_i$ converge to 0 faster. However, with large enough $k$, (\ref{12}) no longer holds, so the designed algorithm is no longer regular. Therefore, it can be seen from subgraphs in the last two columns that the control input $u_i$ has no obvious reset phenomenon. With the increase of $k$ and $h$, the last reset time is shortened as shown in Fig. \ref{fig3ck} and \ref{fig31ch}. Besides, too large $h$ increases the control input of BESS $1$. Fortunately, the increase is not large, and the increase of $h$ will not break the regularity condition. 
\subsection{Case 4: The play-and-plug function}
\begin{figure}
	\centering
	\includegraphics[width=8cm]{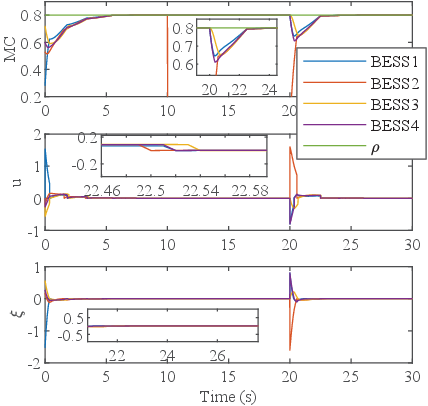} \caption{The evolution of MCs, $u_i$ and $\xi_i$ in Plug-and-Play function test in Case 4.\label{figc4}}
\end{figure}
\begin{figure}
	\centering
	\includegraphics[width=8cm]{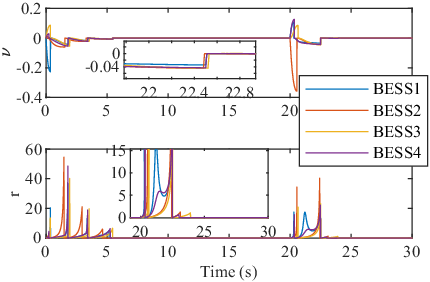} \caption{The evolution of $\nu_i$ and $r_i$ in Plug-and-Play function test in Case 4.\label{figc4c}}
\end{figure}
\par To test the play-and-plug function, BESS 2 is isolated from others at $t=10s$ and reconnected with others at $t=20s$. The test results under a constant electric price are shown in Fig. \ref{figc4}.
\par From the simulation results, when the battery is removed, MCs of other units are obvious affected. At the time of connection restoration, MCs of other units encounters varying degrees of changes, but they can eventually converge to the electricity price.
\subsection{Case 5: Test in time-phased electricity price environment}
\par Time phased electricity pricing is an effective means for EC to guide users in electricity consumption. It can increase electricity prices during peak electricity consumption periods, allowing users to minimize the absorption of power from the power company and even achieve feedback. This reduces the power supply pressure on the power company. Thus, in the time-phased electricity price environment, the designed DED algorithm is verified, and the simulation results are shown in Fig. \ref{figc5}, where the time-phased electricity price scheme is that $\rho=0.8(0\le t\le 15)$, $\rho=1(15\le t\le 35)$ and $\rho=0.6(35\le t\le 50)$. 
\par From the simulation results, it can be seen that, at the instant of electricity price switching, the MC of BESS 1 encounters a larger fluctuation. Nevertheless, MC is still able to converge to the time slot electricity price. Besides, regularity and consensus, as previously analyzed, can be well guaranteed.
\begin{figure}
	\centering
	\includegraphics[width=8cm]{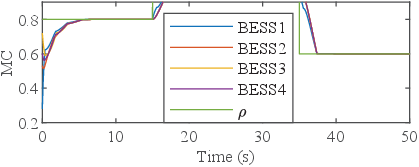} \caption{The evolution of MC under time-phased electricity price in Case 5.\label{figc5}}
\end{figure}
\subsection{Case 6: Test on a Zone-free distributed scheme}
\par In order to verify Theorem \ref{Theorem 3}, the Zone-free distributed scheme with (\ref{19}) is test here, where $\Delta=1$ and $\Delta=1.5$ are introduced into the reset mechanism respectively for comparison. And the simulation results are shown in Fig. \ref{figc6}-\ref{figc61c}.
\par Comparing Fig. \ref{figc6} and \ref{figc61} with \ref{figc1} in Case 1, MCs converge more slowly under the Zone-free distributed reset control scheme, and the larger $\Delta$ is, the slower the convergence speed. In addition, from Fig. \ref{figc6} and \ref{figc61}, the introduction of $\Delta$ causes MCs overshoot, and the larger $\Delta$ is, the greater the overshoot amplitude is. 
\par Comparing Fig. \ref{figc6c} and \ref{figc61c} with \ref{figc1c} in Case 1, in particular, the reset time is delayed, and the larger $\Delta$ is, the more delayed the reset time is. Besides, since the reset operation at some times is ignored, the dynamic gains sometimes appear negative. In other words, $\xi_i$ and $\nu_i$ no longer keep the same number at all times, which is not conducive to the stability and dynamic performance of the system. As described in Remark \ref{remark 5}, it is unwise to select too large $\Delta$, and there is a trade-off between system performance and Zeno-free behavior.
\begin{figure}
	\centering
	\includegraphics[width=8cm]{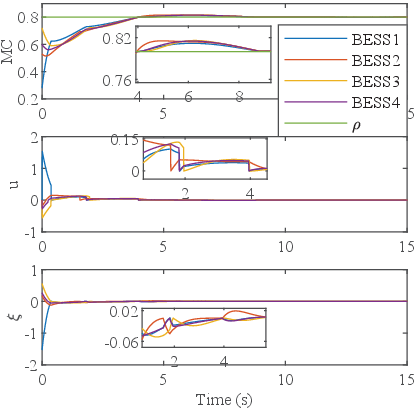} \caption{Simulation results of a Zone-free distributed reset control scheme with $\Delta=1$ in Case 6.\label{figc6}}
\end{figure}
\begin{figure}
	\centering
	\includegraphics[width=8cm]{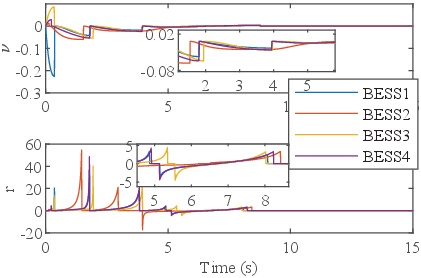} \caption{Simulation results of a Zone-free distributed reset control scheme with $\Delta=1$ in Case 6.\label{figc6c}}
\end{figure}
\begin{figure}
	\centering
	\includegraphics[width=8cm]{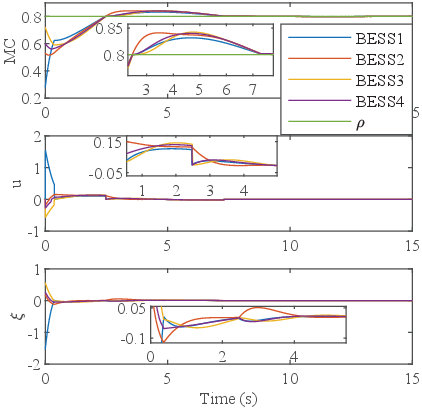} \caption{Simulation results of a Zone-free distributed reset control scheme with $\Delta=1$ in Case 6.\label{figc61}}
\end{figure}
\begin{figure}
	\centering
	\includegraphics[width=8cm]{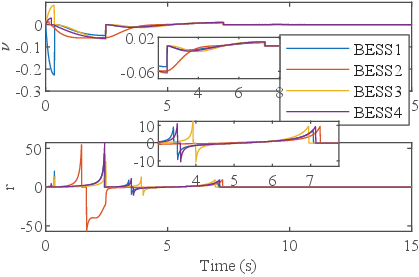} \caption{Simulation results of a Zone-free distributed reset control scheme with $\Delta=1$ in Case 6.\label{figc61c}}
\end{figure}
\subsection{Case 7: Test on a large-scale system}
\begin{figure}
	\centering
	\includegraphics[width=8cm]{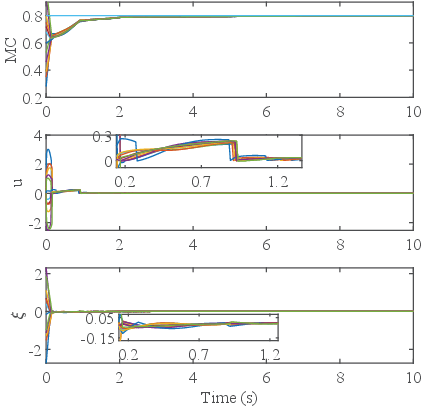} \caption{Simulation results on a large-scale system in Case 7.\label{figc7}}
\end{figure}
\begin{figure}
	\centering
	\includegraphics[width=8cm]{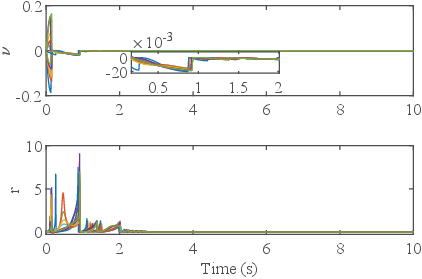} \caption{Simulation results on a large-scale system in Case 7.\label{figc7c}}
\end{figure}
\par To meet the needs of wide area microgrid, the number of BESSs has been increased to 12 based on Case 1, and the same number of agents have been added accordingly. And the simulation results are shown in Fig. \ref{figc7}-\ref{figc7c}.
\par From the simulation results, it can be seen that the designed reset mechanism can cope with the increase in the number of BESS and exhibit the characteristic of integral terms being reset several times. In the end, MCs can converge to the electricity price, while controlling the input, state difference, and integration term to converge to 0. This indicates that the reset mechanism designed is suitable for a large-scale system.
\section{Conclusion}
\par To solve the ED problem in BESS, this paper constructs a composite ED objective convex function containing capacity degradation, power internal loss and electric trading in BESS on a small time-scale so as to obtain MC. And then a MC consensus distributed scheme with reset mechanism is designed. This scheme is an extension of PI controller. Under the action of the reset mechanism, the sign of the integral term in the control input is always aligned with the proportional term. Compared with a PC, MC has the characteristics of fast response and convergence speed, high control accuracy and no overshoot, and no chattering compared with a finite-time controller. The consensus and regularity conditions given guarantee this property. In addition, the Zeno behavior of the reset mechanism is well avoided by introducing a linear reset interval about a minimum dwell time. This paper is the optimization of MC consensus scheme, which can be well transplanted to the ED problem of MG and distributed integrated energy network with convex objective function, respectively.
\appendices

\ifCLASSOPTIONcaptionsoff
  \newpage
\fi
\small
\bibliographystyle{IEEEtran}
\bibliography{re1.bib}
\end{document}